\documentclass[journal]{IEEEtran}

\usepackage{cite}
\usepackage{fixltx2e}
\usepackage[cmex10]{amsmath}
\usepackage{array}
\usepackage{wasysym}
\usepackage{dsfont}
\usepackage[hyperindex]{hyperref}
\usepackage[latin1]{inputenc}                       
\usepackage[english]{babel}                         
\usepackage[T1]{fontenc}
\usepackage{mathtools}
\usepackage{amssymb} 
\usepackage{enumerate}
\usepackage{bbm}
\usepackage{epsfig,syntonly}
\usepackage{verbatim,times}
\usepackage{epstopdf}
\usepackage{graphicx}
\usepackage{accents}
\usepackage{latexsym,fancyhdr,bm}
\usepackage[tight,footnotesize]{subfigure}
\usepackage{url}

\newtheorem{prop}{Proposition}

\makeatletter
\let\l@ENGLISH\l@english
\makeatother

\title{From Polar to Reed-Muller Codes: a Technique to Improve the Finite-Length Performance}

\author{Marco Mondelli, S. Hamed Hassani, and R\"{u}diger Urbanke
\thanks{M. Mondelli and R. Urbanke are with the School of Computer and Communication Sciences,
EPFL, CH-1015 Lausanne, Switzerland
(e-mail: \{marco.mondelli, ruediger.urbanke\}@epfl.ch).

S. H. Hassani is with the Computer Science Department, ETH Z\"{u}rich, Switzerland
(e-mail: hamed@inf.ethz.ch).

This paper was presented in part at the IEEE International Symposium on Information Theory (ISIT), Honolulu, Hawaii, USA, July 2014.}}

\begin{document}

\maketitle
\begin{abstract}
\noindent We explore the relationship between polar and RM codes and we describe a coding scheme which improves upon the performance of the standard polar code at practical block lengths. Our starting point is the experimental observation that RM codes have a smaller error probability than polar codes under MAP decoding. This motivates us to introduce a family of codes that ``interpolates'' between RM and polar codes, call this family ${\mathcal C}_{\rm inter} = \{C_{\alpha} : \alpha \in [0, 1]\}$, where $C_{\alpha} \big |_{\alpha = 1}$ is the original polar code, and $C_{\alpha} \big |_{\alpha = 0}$ is an RM code. Based on numerical observations, we remark that the error probability under MAP decoding is an increasing function of $\alpha$. MAP decoding has in general exponential complexity, but empirically the performance of polar codes at finite block lengths is boosted by moving along the family ${\mathcal C}_{\rm inter}$ even under low-complexity decoding schemes such as, for instance, belief propagation or successive cancellation list decoder. We demonstrate the performance gain via numerical simulations for transmission over the erasure channel as well as the Gaussian channel.
\end{abstract}


\section{Introduction} \label{sec:intro}

{\bf \em Polar Coding: Benefits and Drawbacks.} Polar codes, which were
introduced by Ar{\i}kan in \cite{arikan:polar}, are a family of codes which provably achieve the capacity of a large class of channels, including binary-input memoryless output-symmetric channels (BMSCs), by means of
encoding and decoding algorithms with complexity $\Theta(N \log
N)$, $N$ being the block length of the code.

In particular, for any BMSC $W$ with capacity $I(W)$ and for any rate $R < I(W)$, the block error probability
under the proposed successive cancellation (SC) decoding,
namely $P_{\rm e}^{\rm SC}$, scales roughly as $2^{-\sqrt{N}}$
as $N$ grows large \cite{arikan:rate}. This result has been further refined and extended
to the MAP decoder, showing that both $\log_2(-\log_2 P_{\rm e}^{\rm
SC})$ and $\log_2(-\log_2 P_{\rm e}^{\rm MAP})$ behave as $\log_2 (N)/2+\sqrt{\log_2 (N)}/2\cdot
Q^{-1}(R/I(W)) + o(\sqrt{\log_2 (N)})$ for any fixed rate strictly less than
capacity \cite{hassani:scalingI, mori:rate}. Consequently, even at moderate
block lengths, error floors do not affect the performance of polar codes.

However, when we consider rates close to capacity, simulation
results show that large block lengths are required in order to achieve
a desired error probability. Therefore, it is interesting to
explore the trade-off between the gap to capacity $I(W)-R$ and the block length $N$ when the error
probability is a fixed value $P_{\rm e}$. In particular, it has been observed that $I(W)-R$ scales as $N^{-1/\mu}$, where $\mu$ denotes the scaling exponent \cite{korada:scaling}. Note that, in general, the scaling exponent is not related to the error exponent, since they concern two different regimes: for the scaling exponent, we fix the error probability and study the scaling of the gap to capacity with respect to the block length; for the error exponent, we fix the rate and study the scaling of the error probability with respect to the block length. For transmission over the binary erasure channel (BEC), an estimation for the scaling exponent is known, namely $\mu \approx 3.627$. Therefore, compared to random codes which have a scaling exponent of $2$, polar codes require larger block lengths to achieve the same rate and error probability. For a generic BMSC, taking as a proxy of the error probability
the sum of the Bhattacharyya parameters, the scaling
exponent is lower bounded by 3.553 \cite{HAU13} and upper bounded
by 5.77 \cite{GB13}. Furthermore, it is conjectured that the lower
bound on $\mu$ can be increased up to 3.627, namely, to the value for the BEC.

In order to improve the finite-length performance of polar codes, several decoding algorithms have been proposed. Maximum likelihood (ML) decoders are implemented via the Viterbi algorithm \cite{arikan:mlshort} and via sphere decoding \cite{kahraman:ml}, but are practical only for relatively short block lengths. A linear programming (LP) decoder is introduced in \cite{korada:lp},  and the performance under belief propagation (BP) decoding is considered in \cite{hussami:perfo}. The stopping set analysis for transmission over the BEC is also provided in \cite{eslami:finite}. A successive cancellation list (SCL) decoder is proposed in
\cite{vardy:listpolar}. Empirically, the usage of $L$ concurrent
decoding paths yields a significant improvement in the achievable
error probability and allows to obtain an error probability comparable to that under MAP decoding with practical values of the list size. However, it has been recently shown that, under MAP decoding, the introduction of any finite list does not change the scaling exponent \cite{mondelli:scalingITW}. In particular, for any BMSC and for any family of linear codes with unbounded minimum distance, list decoding cannot modify the scaling behavior for finite values of the list size. Analogously, under genie-aided SC decoding, the scaling exponent stays constant for any fixed number of helps from the genie, when transmission takes place over the BEC.

{\bf \em Reed-Muller Codes and Their Relation to Polar Coding.} RM codes were introduced by Muller \cite{muller:code} and rediscovered shortly thereafter with an efficient decoding algorithm by Reed \cite{reed:code}. The relation between polar codes and RM codes was first pointed out in \cite{arikan:polar}. Performance comparisons were carried out in \cite{arikan:rmpolarcomp, arikan:rmsurvey}. It was observed in \cite{korada:thesis} that Dumer's recursive algorithm for RM codes \cite{dumer:RMalgo} is similar to the SC decoder for polar codes. In addition, list decoding has been used also to improve the performance of RM codes \cite{dumer:RMlist, dumer:recl}. Furthermore, recursive techniques can be employed to decode nested polarized codes in which the splitting process ends at various short RM codes instead
of the single information bits used as end nodes in polar
codes \cite{dumer:manu, dumer:manu2}. Numerical simulations and analytical results suggest that RM codes have a bad performance under successive and iterative decoding, but they outperform polar codes under MAP decoding \cite{arikan:polar, hussami:perfo}. Indeed, RM codes have better minimum distance properties and an hybrid design which combines the construction of RM and polar codes is introduced in \cite{tse:manu}. However, no rigorous results are known and the fundamental problem concerning whether RM codes are capacity-achieving under MAP decoding, at least for some channels with a sufficient amount of symmetry, remains open \cite{costello:roadcapacity}. 

{\bf \em Contribution of the Present Work.} In this paper we propose an interpolation method between the polar code of block length $N$ and rate $R$ and an RM code of the same block length and rate. To do so, we describe a family of codes ${\mathcal C}_{\rm inter} = \{C_{\alpha} : \alpha \in [0, 1]\}$ such that $C_{\alpha} \big |_{\alpha = 1}$ is the original polar code, and $C_{\alpha} \big |_{\alpha = 0}$ is an RM code. We remark that experimentally the error probability under MAP decoding increases with $\alpha$ for transmission over the BEC and over the binary additive white Gaussian noise channel (BAWGNC). Even if MAP decoding is in general an NP-complete task, this result is relevant in practice because picking suitable codes from ${\mathcal C}_{\rm inter}$ boosts the finite length performance of the original polar code also when low-complexity suboptimal algorithms are employed. In particular, a remarkable performance improvement is noticed adopting the SCL decoder proposed in \cite{vardy:listpolar} and the BP decoder. This performance gain could be substantial in the sense of the reduction of the scaling exponent: according to numerical simulations performed for $N=2^{10}$ over the BEC, the error probability under MAP decoding for the transmission of $C_{\alpha}$ for $\alpha$ sufficiently small is very close to that of random codes. As a result, the usage of codes in ${\mathcal C}_{\rm inter}$ potentially improves the speed at which capacity is reached.

{\bf \em Organization.} Section \ref{sec:interp} points out similarities and differences between the polar and the RM construction and describes explicitly the interpolating family ${\mathcal C}_{\rm inter}$ for the special case of the transmission over the BEC. Starting from the analysis of the two extreme cases of MAP and SC decoding, Section \ref{sec:perfoBEC} shows how to improve significantly the finite-length performance of polar codes by using codes of the form $C_{\alpha}$ decoded with low-complexity suboptimal schemes when transmission takes place over the BEC. The interpolation method between RM and polar codes is described for the transmission over a generic BMSC $W$ in Section \ref{sec:perfoBMSC}, where the simulation results for the BAWGNC are presented as a case study. Finally, Section \ref{sec:concl} draws the conclusions of the paper.

\section{From Polar to RM Codes: an Interpolation Method for the BEC} \label{sec:interp}

Let $n \in {\mathbb N}$ and $N=2^n$. Consider the $N \times N$ matrix $G_N$ defined as follows,
\begin{equation}
G_N = F^{\otimes n}, \qquad \qquad F =  \biggl[ \begin{array}{cc}
1 & 0 \\
1 & 1 \end{array} \biggr], 
\end{equation}
where $F^{\otimes n}$ denotes the $n$-th Kronecker power of $F$. As it has been formerly pointed out in \cite{arikan:polar}, the generator matrices of both polar and RM codes are obtained by suitably selecting rows from $G_N = (g_1, \cdots, g_N)^T$. 

In particular, the \emph{RM rule} for building a code of block length $N$ and minimum distance $2^k$ for some fixed $k \in \{0, 1, \cdots, n\}$ consists in choosing the rows of $G_N$ with Hamming weight at least $2^k$. Thus, the rate $R$ of this code is given by
\begin{equation} \label{eq:RMrate}
R = \frac{\displaystyle\sum_{i=k}^n \binom{n}{i}}{N}.
\end{equation}
In general, if we require an RM code with fixed block length $N$ and rate $R$, where $R$ cannot be written in the form \eqref{eq:RMrate} for some $k \in {\mathbb N}$, we take as generator matrix any subset of $NR$ rows of $G_N$ with the highest Hamming weights. Notice that this criterion is channel-independent in the sense that it does not rely on the particular channel over which the transmission takes place. 

On the other hand, the \emph{polar rule} is channel-specific. Indeed, the $N$ synthetic channels $W_N^{(i)}$ ($i \in \{0, \cdots, N-1\}$) are obtained from $N$ independent copies of the original channel $W$. The row $g_i$ is associated to $W_N^{(i)}$ and the synthetic channels (and, therefore, the rows) with the lowest Bhattacharyya parameters\footnote{The Bhattacharyya parameter $Z_i$ of the synthetic channel $W_N^{(i)}$ represents a measure of the reliability of the channel: $Z_i$ is close to 0 or to 1 if and only if the entropy of the $i$-th position given the previous $i-1$ bits is close to 0 or to 1, respectively. Hence, if $Z_i$ is close to 0, then the $i$-th position can be decoded with high probability given the previous $i-1$ bits, while if $Z_i$ is close to 1, the decoding fails with high probability.} are selected. In general, different channels $W$ yield different choices of rows. Let us consider the simple case of the transmission over the binary erasure channel with erasure probability $\varepsilon$, in short BEC$(\varepsilon)$, for fixed $\varepsilon \in (0, 1)$. In this particular scenario, the Bhattacharyya parameter $Z_i$ associated to $W_N^{(i)}$ (and, therefore, to $g_i$) is given by
\begin{equation} \label{eq:Bhatt}
Z_i(\varepsilon) = f_{b_1^{(i)}}\circ f_{b_2^{(i)}}\circ \cdots f_{b_n^{(i)}}(\varepsilon),
\end{equation}  
where $f_0(x) = 1-(1-x)^2$, $f_1(x) = x^2$, $\circ$ denotes function composition, and $b^{(i)}=(b_1^{(i)}, b_2^{(i)}, \cdots, b_n^{(i)})^T$ is the binary expansion of $i$ over $n$ bits, $b_1^{(i)}$ being the most significant bit and $b_n^{(i)}$ the least significant bit. In order to construct a code of block length $N$ and rate $R$, we select the $NR$ rows which minimize the expression \eqref{eq:Bhatt}. 

The link between the RM rule and the polar rule is clarified by the following proposition.

\begin{prop} \label{prop:RMP}
The polar code of block length $N$ and rate $R$ designed for transmission over a BEC$(\varepsilon)$, when $\varepsilon \to 0$, is an RM code.
\end{prop}

\begin{proof}
Suppose that the thesis is false, i.e., that we include  $g_{j^*}$, but not $g_{i^*}$, with $w_{\rm H}(g_{i^*}) > w_{\rm H}(g_{j^*})$, where $w_{\rm H}(\cdot)$ denotes the Hamming weight. Since $w_{\rm H}(g_i) =2^{\sum_{k=1}^n b_k^{(i)}} = 2^{w_{\rm H}(b^{(i)})}$ for any $i \in \{0, \cdots, N-1\}$ (Proposition 17 of \cite{arikan:polar}), then $w_{\rm H}(b^{(i^*)}) > w_{\rm H}(b^{(j^*)})$.

From formula \eqref{eq:Bhatt}, one deduces that $Z_i(\varepsilon)$ is a polynomial in $\varepsilon$ with minimum degree equal to $2^{w_{\rm H}(b^{(i)})}$. Hence,
\begin{equation*}
\lim_{\varepsilon \to 0} \frac{Z_{i^*}(\varepsilon)}{Z_{j^*}(\varepsilon)}=0,
\end{equation*}
which means that there exists $\delta >0$ s.t. for all $\varepsilon < \delta$, $Z_{i^*}(\varepsilon) < Z_{j^*}(\varepsilon)$. Consequently, a polar code designed for transmission over a BEC$(\varepsilon)$, with $\varepsilon < \delta$, which includes $g_{j^*}$ must also include $g_{i^*}$. This is a contradiction.
\end{proof}

Recall that the transmission takes place over $W =$ BEC$(\varepsilon)$. Let $C_{\alpha}$ be the polar code of block length $N$ and rate $R$ designed for a BEC$(\alpha \varepsilon)$. When $\alpha = 1$, $C_{\alpha}$ reduces to the polar code for the channel $W$, while, when $\alpha \to 0$, $C_{\alpha}$ becomes an RM code by Proposition \ref{prop:RMP}. Consider the family of codes ${\mathcal C}_{\rm inter}$ defined as,
\begin{equation} \label{eq:Cinterp}
{\mathcal C}_{\rm inter} = \{C_{\alpha} : \alpha \in [0, 1]\}.
\end{equation}
The codes in ${\mathcal C}_{\rm inter}$ provide an \emph{interpolation method} to pass smoothly from a polar code to an RM code of the same rate and block length. Indeed, consider the generator matrices of the codes in ${\mathcal C}_{\rm inter}$ which are obtained reducing $\alpha$ from $1$ to $0$. We start from the generator matrix of the polar code and the successive matrices are obtained by changing one row at a time. In particular, numerical simulations show that the row which is included in the next code (associated to a smaller $\alpha$) has a higher Hamming weight than the row which was removed from the previous code (associated to a higher $\alpha$). Heuristically, this happens for the following reason. The row indices chosen by $C_{\alpha}$ are the ones which minimize the associated Bhattacharyya parameters $Z_i(\alpha \varepsilon)$ given by \eqref{eq:Bhatt}. As $f_1(x)\le f_0(x)$ for any $x \in[0, 1]$, applying $f_1$ instead of $f_0$ makes the Bhattacharyya parameter decrease. However, also the order in which the functions are applied is important, since $f_0 \circ f_1(x) \le f_1 \circ f_0(x)$ for any $x \in [0, 1]$: if we fix $w_{\rm H}(b^{(i)})$, $Z_i$ is minimized by applying first all the functions $f_1$ and then the functions $f_0$. Therefore, the goodness of the index $i$ depends both on the number of $1$'s in its binary expansion $b^{(i)}$ and on the positions of these $1$'s. On the other hand, when designing an RM code only $w_{\rm H}(b^{(i)})$ matters and, for $\alpha$ small enough, $C_{\alpha}$ tends to an RM code. As a result, as $\alpha$ goes from $1$ to $0$, the value of $Z_i(\alpha \varepsilon)$ depends more and more on $w_{\rm H}(b^{(i)})$ than on the position of the $1$'s in $b^{(i)}$.

\section{Improving the Finite-Length Performance of Polar Codes for the BEC} \label{sec:perfoBEC}

The focus of this section is on the performance of the codes in ${\mathcal C}_{\rm inter}$ when transmission takes place over the BEC$(\varepsilon)$. We start considering the MAP decoder and then move to the SC decoder introduced by Ar{\i}kan. By taking into account low-complexity suboptimal decoding schemes which outperform the original SC algorithm (e.g., SCL and BP), we highlight the advantage of employing codes of the form $C_{\alpha}$. The simulation results of this section refer to codes of fixed block length $N= 2^{10}$ and rate $R=0.5$. The number of Monte Carlo trials is $M = 10^5$.

\subsection{Motivation: MAP Decoding}

Since it has been observed that under MAP decoding picking the rows of $G_N$ according to the RM rule significantly improves the performance with respect to the polar choice \cite{hussami:perfo}, it is interesting to analyze the error probability $P_{\rm e}^{\rm MAP}(\alpha, \varepsilon)$ under MAP decoding for the transmission of the code $C_{\alpha}$ over the BEC$(\varepsilon)$. Although MAP decoding is in general an NP-complete task, for the particular case of the BEC it is equivalent to the inversion of a suitable matrix and, therefore, can be performed in ${\mathcal O}(N^3)$.

\begin{figure}[t] 
\centering 
\includegraphics[width=0.86\columnwidth]{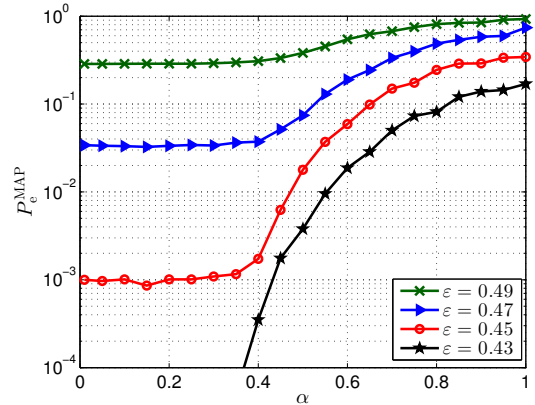}
\caption{Error probability $P_{\rm e}^{\rm MAP}$ under MAP decoding for the transmission of $C_{\alpha}$ over the BEC$(\varepsilon)$, when $\alpha$ varies from $0$ to $1$ with a step of $0.05$ and $\varepsilon$ is given four distinct values. The block length is $N=2^{10}$ and the rate is $R=0.5$. Observe that $P_{\rm e}^{\rm MAP}$ is increasing in $\alpha$ for all values of $\varepsilon$, which means that the minimum error probability is achieved by the RM code $C_{\alpha}\big|_{\alpha = 0}$.} 
\label{fig:MAPepfix}
\end{figure} 

First of all, fix the value of $\varepsilon$ and consider how $P_{\rm e}^{\rm MAP}$ varies as a function of $\alpha$. As it is shown in Figure \ref{fig:MAPepfix} for four distinct values of $\varepsilon$, $P_{\rm e}^{\rm MAP}(\alpha, \varepsilon)$ is increasing in $\alpha$. In short, the proposed interpolation method to pass from the polar code $C_{\alpha} \big|_{\alpha = 1}$ to an RM code $C_{\alpha} \big|_{\alpha = 0}$ yields a family of codes with \emph{decreasing} MAP error probability. This conjecture, if proved, would imply that RM codes are capacity-achieving for the BEC, which is a long-standing open problem in coding theory. Another evidence in support of this statement is as follows. As it has been pointed out in Section \ref{sec:interp}, the polar rule differs from the RM rule in the fact that not only the number, but also the position of the $1$'s in $b^{(i)}$ matters in the choice of the row indices. In particular, polar codes prefer to set the $1$'s in the least significant bits of the binary expansion of $i$. However, if one is concerned with achieving the capacity of the BEC under MAP decoding, the specific order of the $1$'s in the binary expansions of the row indices does not play any role. Indeed, denote by ${\mathcal F}$ the set of row indices of $G_N$ which are not chosen for the generator matrix of the polar code (these indices are \emph{frozen}, since they are not used for the transmission of information bits) and let ${\mathcal F}^c$ be its complement. Then, it is possible to arbitrarily permute the binary expansions $b^{(i)}$ ($i \in {\mathcal F}^c$) and still get a set of row indices which yields a capacity-achieving family of codes under MAP decoding. This fact is formalized in the following proposition.

\begin{prop} \label{prop:overcomplete}
Denote by ${\mathcal F}^c$ the set of row indices chosen by polar coding. Let $\pi : \{1, \cdots, n\} \to \{1, \cdots, n\}$ be a permutation and let $P_{\pi}$ be the associated permutation matrix. Construct the code $C_{\pi}$ by taking the rows of $G_N$ whose indices have binary expansions $P_{\pi} b^{(i)}$  for $i \in {\mathcal F}^c$. Let $\varepsilon \in (0, 1)$ and denote by $P_{\rm e}^{\mathcal D}(C_{\pi})$ the error probability under the decoder $\mathcal D$ for the transmission of $C_{\pi}$ over the BEC$(\varepsilon)$. Then, $P_{\rm e}^{\rm MAP}(C_{\pi}) \le P_{\rm e}^{\rm SC}(C_{\iota})$, $C_{\iota}$ being the original polar code.
\end{prop}

\begin{proof}
As observed in \cite{hussami:perfo}, there exist $n!$ different representations of the polar code $C_{\iota}$ of block length $N=2^n$ obtained by permuting the $n$ layers of connections. Let us apply the permutation $\tau$ to these layers and then run the SC algorithm, denoting by $P_{\rm e}^{{\rm SC}, \tau}(C_{\iota})$ the error probability for transmission over the BEC$(\varepsilon)$. The application of the permutation $\tau$ affects the Bhattacharyya parameter $Z_i$ associated to the synthetic channel $W_N^{(i)}$, which is now given by
\begin{equation*}
Z_i(\varepsilon) = f_{\tau(b_1^{(i)})}\circ f_{\tau(b_2^{(i)})}\circ \cdots f_{\tau(b_n^{(i)})}(\varepsilon).
\end{equation*}
On the other hand, the generator matrix (and, consequently, the set ${\mathcal F}^c$) does not change, because the code stays the same. Therefore, the probability that the SC decoder fails when applying the permutation $\tau$ to the layers of the code $C_{\iota}$ equals the probability that the SC decoder fails when the code $C_{\tau}$ is employed. In formulas, for any permutation $\tau$,
\begin{equation*}
P_{\rm e}^{{\rm SC}, \tau}(C_{\iota}) = P_{\rm e}^{\rm SC}(C_{\tau}).
\end{equation*}

Denote by OSC the algorithm which runs SC decoding over all the $n!$ possible overcomplete representation of a polar code. When transmission takes place over the BEC, the OSC decoder fails if and only if there exists an information bit which cannot be decoded by any of these $n!$ SC decoders. Let $P_{\rm e}^{\rm OSC}(C_{\pi})$ be the error probability under OSC decoding for transmission of the code $C_{\pi}$ over the BEC$(\varepsilon)$. Then, $P_{\rm e}^{\rm OSC}(C_{\pi}) \le P_{\rm e}^{{\rm SC}, \tau}(C_{\pi})$ for any $\tau$. Taking $\tau = \pi^{-1}$ and recalling that MAP decoding minimizes the error probability, we obtain that 
\begin{equation*}
P_{\rm e}^{\rm MAP}(C_{\pi}) \le P_{\rm e}^{\rm OSC}(C_{\pi}) \le P_{\rm e}^{{\rm SC}, \pi^{-1}}(C_{\pi}) = P_{\rm e}^{\rm SC}(C_{\iota}),
\end{equation*}
which gives us the desired result.
\end{proof}

In Figure \ref{fig:MAPalphafix} we fix the value of $\alpha$ and we analyze $P_{\rm e}^{\rm MAP}$ as a function of $\varepsilon$. It is interesting to remark that already for $\alpha = 0.3$, the error probability for the transmission of $C_{\alpha}$ is very close to that of random coding, which not only achieves capacity, but does so with a more favorable tradeoff between $N$ and $I(W)-R$. Indeed, random codes have a scaling exponent $\mu =2$, while the scaling exponent of polar codes is $\mu = 3.627$\footnote{Note that there is no conflict between the facts that (i) the error exponent of RM codes under MAP decoding cannot be as good as that of random codes because of their minimum distance \cite{hussami:perfo} and (ii) the scaling exponent of RM codes can match that of random codes. Indeed, the error exponent and the scaling exponent concern two different limits. For example, an error probability of the form $2^{-a\sqrt{N}}+ 2^{-b N(C-R)^2}$ for some constants $a$ and $b$ yields the error exponent of polar codes and, at the same time, the scaling exponent of random codes.}. 

\begin{figure}[t] 
\centering 
\includegraphics[width=0.86\columnwidth]{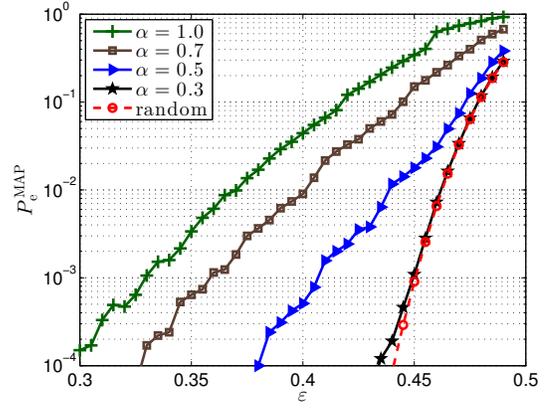}
\caption{Error probability $P_{\rm e}^{\rm MAP}$ under MAP decoding for the transmission of $C_{\alpha}$ over the BEC$(\varepsilon)$, when $\varepsilon$ varies from $0.30$ to $0.49$ with a step of $0.005$ and $\alpha$ is given four distinct values. The block length is $N=2^{10}$ and the rate is $R=0.5$. Remark that already for $\alpha =0.3$ the error performance of $C_{\alpha}$ is comparable to that of random codes.} 
\label{fig:MAPalphafix}
\end{figure}

\subsection{SC Decoding}

After dealing with optimal MAP decoding, let us analyze the performance of the codes in ${\mathcal C}_{\rm inter}$ under SC decoding. As can be seen in Figure \ref{fig:SCepfix} for four distinct values of $\varepsilon$, the error probability $P_{\rm e}^{\rm SC}(\alpha, \varepsilon)$ under SC decoding for transmission of the code $C_{\alpha}$ over the BEC$(\varepsilon)$ is a decreasing function of $\alpha$. Hence, the best performance are obtained using the polar code $C_{\alpha}\big |_{\alpha =1}$. The theoretical reason of this behavior lies in the fact that $P_{\rm e}^{\rm SC}$ can be well approximated by the sum of the Bhattacharyya parameters of the synthetic channels which are selected by the polar code for transmission of the information bits \cite{mani:correlation}. Formally, let ${\mathcal F}^c(\alpha)$ be the set of indices which are selected by the polar code $C_{\alpha}$. Then,
\begin{equation} \label{eq:unionbound}
P_{\rm e}^{\rm SC}(\alpha) \apprle \sum_{i \in {\mathcal F}^c(\alpha)}Z_i(\varepsilon).
\end{equation}
The bound \eqref{eq:unionbound} is tight and $\sum_{i \in {\mathcal F}^c(\alpha)}Z_i(\varepsilon)$ is minimized for $\alpha = 1$.


\begin{figure}[t] 
\centering 
\includegraphics[width=0.86\columnwidth]{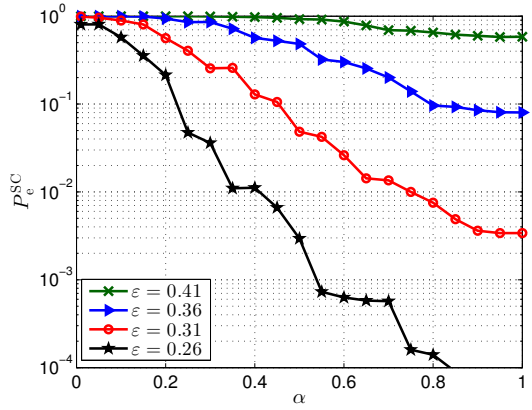}
\caption{Error probability $P_{\rm e}^{\rm SC}$ under SC decoding for the transmission of $C_{\alpha}$ over the BEC$(\varepsilon)$, when $\alpha$ varies from $0$ to $1$ with a step of $0.05$ and $\varepsilon$ is given four distinct values. The block length is $N=2^{10}$ and the rate is $R=0.5$. Observe that $P_{\rm e}^{\rm SC}$ is decreasing in $\alpha$, which means that the minimum $P_{\rm e}^{\rm SC}$ is achieved by the original polar code $C_{\alpha}\big|_{\alpha=1}$.} 
\label{fig:SCepfix}
\end{figure}


\subsection{Something Between the Two Extremes: List Decoding and Belief Propagation}


\begin{figure}[t]
\centering
\subfigure[$\alpha =0.9$]
{\includegraphics[width=0.86\columnwidth]{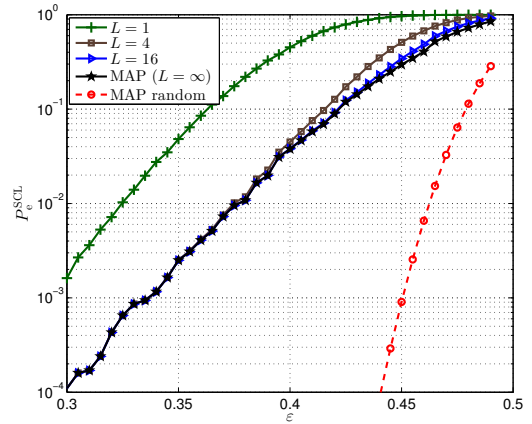}}\\
\subfigure[$\alpha=0.4$]
{\includegraphics[width=0.86\columnwidth]{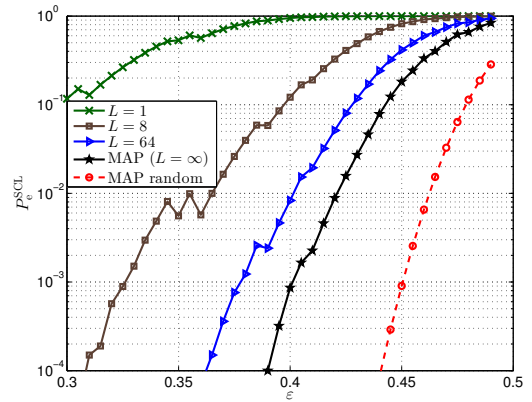}}
\caption{Error probability $P_{\rm e}^{\rm SCL}$ under SCL decoding for the transmission of $C_{\alpha}$ over the BEC$(\varepsilon)$ for different values of the list size $L$, when $\varepsilon$ varies from $0.30$ to $0.49$ with a step of $0.005$. The block length is $N=2^{10}$ and the rate is $R=0.5$. As a benchmark, we represent also the error probability under MAP decoding for the transmission of $C_{\alpha}$ (in black) and for the transmission of a random code (in red). Observe that if $\alpha$ is big (upper plot), $P_{\rm e}^{\rm SCL}$ converges to $P_{\rm e}^{\rm MAP}$ already for small values of the list size. On the other hand, if $\alpha$ is small (lower plot), bigger list sizes are required to get to the error probability of MAP decoding, which in return becomes much smaller in value and, therefore, much closer to the error probability of a random code. The fact that some curves are not always increasing in $\varepsilon$ is not caused by a problem in the simulation. Indeed, the code changes with $\varepsilon$ and, for a small variation of the channel parameter, this can lead to such unexpected effects, which can be noticed also in Figures \ref{fig:SCLdiffL} and \ref{fig:BPalphafix}.}
\label{fig:SCLdiffalpha}
\end{figure}

Consider the SCL scheme introduced in \cite{vardy:listpolar} and denote by $P_{\rm e}^{\rm SCL}(\alpha, \varepsilon, L)$ the error probability under SCL decoding with list size $L$ for transmission of the polar code $C_{\alpha}$ over the BEC$(\varepsilon)$. 
Clearly, if $L=1$, this scheme reduces to the SC algorithm originally proposed by Ar{\i}kan, while for $L \ge 2^{NR}$, the SCL decoder is equivalent to the MAP decoder, since the list is big enough to contain all the possible $2^{NR}$ codewords. Therefore, as $L$ increases, we gradually pass from SC decoding to MAP decoding. 

If we fix $\alpha$ and we let $L$ grow, $P_{\rm e}^{\rm SCL}(\alpha, \varepsilon, L)$ monotonically decreases from $P_{\rm e}^{\rm SC}(\alpha, \varepsilon)$ to $P_{\rm e}^{\rm MAP}(\alpha, \varepsilon)$. Recall that, as $\alpha$ goes from $1$ to $0$, $P_{\rm e}^{\rm SC}(\alpha, \varepsilon)$ increases, while $P_{\rm e}^{\rm MAP}(\alpha, \varepsilon)$ decreases. Values of $\alpha$ close to $1$ imply that $P_{\rm e}^{\rm SCL}(\alpha, \varepsilon, L)$ gets close to the MAP error probability for small values of the list size. If $\alpha$ is reduced, a bigger list size is required to obtain performance comparable to MAP decoding since the underlying SC algorithm gets worse, but $P_{\rm e}^{\rm MAP}(\alpha, \varepsilon)$ becomes significantly smaller. In other words, a smaller $\alpha$ implies a slower converge (in terms of $L$) toward a smaller error probability. This trade-off between MAP error probability and list size required to reach it is illustrated in Figure \ref{fig:SCLdiffalpha} for $\alpha =0.9$ and $\alpha =0.4$, where, as a benchmark, we represent also the average error probability under MAP decoding for the transmission of random codes. 

\begin{figure}[t]
\centering
\subfigure[$L=8$]
{\includegraphics[width=0.86\columnwidth]{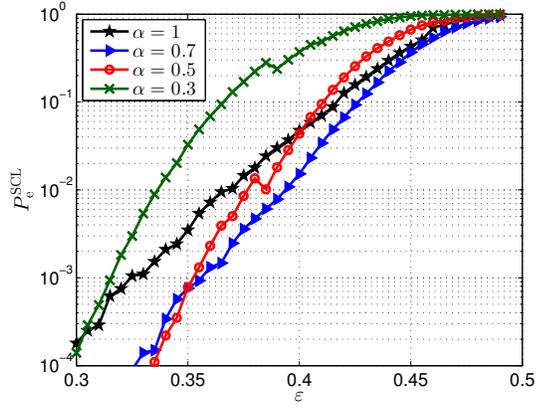}}\\
\subfigure[$L=32$]
{\includegraphics[width=0.86\columnwidth]{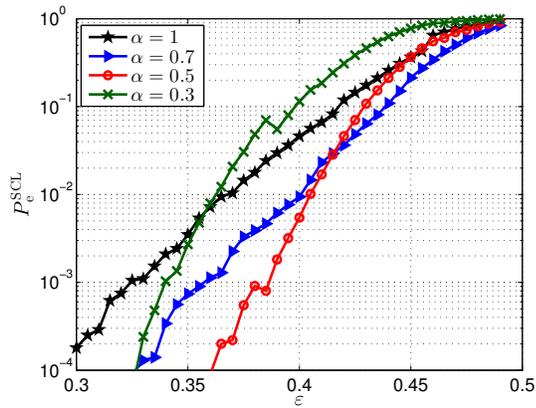}}
\caption{Error probability $P_{\rm e}^{\rm SCL}$ under SCL decoding for the transmission of $C_{\alpha}$ over the BEC$(\varepsilon)$, when $\varepsilon$ varies from $0.30$ to $0.49$ with a step of $0.005$ and for different values of $\alpha$. The block length is $N=2^{10}$ and the rate is $R=0.5$. Already when $L=8$ (upper plot), a performance improvement is obtained reducing $\alpha$ with respect to the original polar code $C_{\alpha}\big|_{\alpha=1}$. If the list size is increased to $L=32$ (lower plot), the advantage in considering codes $C_{\alpha}$ with a smaller value of the tuning parameter $\alpha$ is even more evident.}
\label{fig:SCLdiffL}
\end{figure}

In order to show that the usage of codes in ${\mathcal C}_{\rm inter}$ significantly improves the finite-length performance of polar codes for practical values of the list size, fix $L$ and consider the transmission of $C_{\alpha}$ for different values of $\alpha$. The results for $L=8$ and $L=32$ are represented in Figure \ref{fig:SCLdiffL}. The code $C_{\alpha}\big|_{\alpha=0.7}$ outperforms the original polar scheme already when $L=8$. If the decoder is allowed to take $L=32$, the improvement in performance is even more significant and, for example, the target error probability $P_{\rm e} =10^{-3}$ can be obtained for $\varepsilon = 0.385$ if we employ $C_{\alpha}\big|_{\alpha=0.5}$, while $\varepsilon = 0.325$ is required if we employ the original polar code $C_{\alpha}\big|_{\alpha=1}$. Remark that if the target error probability to be met is very low, it is convenient to consider codes $C_{\alpha}$ with small $\alpha$, since they will be able to achieve it for higher erasure probabilities of the BEC. Indeed, observe that in the case $L=32$, $C_{\alpha}\big|_{\alpha=0.3}$ outperforms the original polar code for $P_{\rm e}^{\rm SCL} < 10^{-3}$. This effect is due to the fact that, for any fixed rate less than capacity, $P_{\rm e}^{\rm SC}$ scales with $N$ as $2^{-\sqrt{N}}$ and, hence, polar codes are not affected by error floors.

In general, it is convenient to consider codes of the form $C_{\alpha}$ whenever the decoding algorithm yields better results than the SC decoder. As another example, consider the case of the BP decoder. It has been already pointed out that the polar choice of the row indices to be selected from $G_N$ is not optimal for the BP algorithm \cite{hussami:perfo, eslami:finite}, but no systematic rule capable of outperforming polar codes is known. As can be seen in Figure \ref{fig:BPalphafix}, the interpolating family ${\mathcal C}_{\rm inter}$ contains codes which achieve a smaller error probability than that of the original polar code $C_{\alpha} \big |_{\alpha = 1}$ for an appropriate choice of the parameter $\alpha$.

\begin{figure}[t] 
\centering 
\includegraphics[width=0.86\columnwidth]{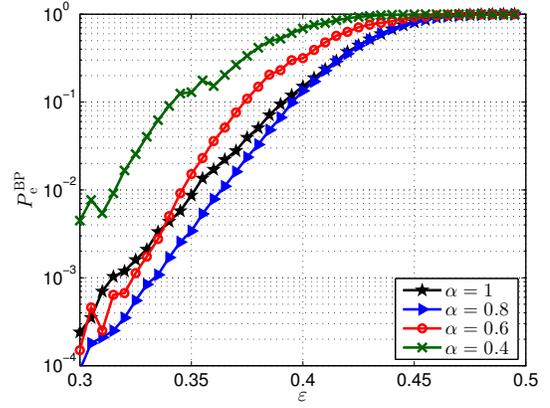}
\caption{Error probability $P_{\rm e}^{\rm BP}$ under BP decoding for the transmission of $C_{\alpha}$ over the BEC$(\varepsilon)$, when $\varepsilon$ varies from $0.30$ to $0.49$ with a step of $0.005$ and $\alpha$ is given four distinct values. The block length is $N=2^{10}$ and the rate is $R=0.5$. Remark that the optimal performance is obtained with the code $C_{\alpha} \big |_{\alpha = 0.8}$.} 
\label{fig:BPalphafix}
\end{figure}

\section{Generalization to Any BMSC} \label{sec:perfoBMSC}

This section is devoted to the generalization of the ideas expressed for the BEC in Sections \ref{sec:interp} and \ref{sec:perfoBEC} to the transmission over a BMSC $W$. In particular, first we propose a method for constructing the family of codes ${\mathcal C}_{\rm inter}$ and, then, we analyze the performance for the transmission over a BAWGNC.

\subsection{General Construction of an Interpolating Family}

Suppose that the transmission takes place over the BMSC $W$ and let $Z(W)$ be its Bhattacharyya parameter. In order to construct the interpolating family ${\mathcal C}_{\rm inter}$, we consider the family of channels ${\mathcal W}_{\rm inter}$ ordered by degradation \cite{urbanke:coding} such that the element of the family with the biggest Bhattacharyya parameter is $W$ itself and the element of the family with the smallest Bhattacharyya parameter is the perfect channel $W^{\rm opt}$, in which the output is equal to the input with probability $1$. There are many ways of performing such a task. In particular, we can set
\begin{equation}
{\mathcal W}_{\rm inter} = \{W_{\alpha} : \alpha \in [0, 1]\},
\end{equation}
where $W_{\alpha} = W$ with probability $\alpha$, $W_{\alpha} = W^{\rm opt}$ with probability $1-\alpha$, and the receiver knows which channel has been used. In formulas, $W_{\alpha} = \alpha W + (1-\alpha) W^{\rm opt}$.

Since the convex combination of BMS channels is a BMS channel, $W_{\alpha}$ is also a BMSC with Bhattacharyya parameter $Z_{\alpha} = \alpha Z$. Denote by $C_{\alpha}$ the polar code for transmission over $W_{\alpha}$. Then, the interpolating family ${\mathcal C}_{\rm inter}$ is defined as in \eqref{eq:Cinterp}. This is a reasonable choice for ${\mathcal C}_{\rm inter}$ because of the following result, which extends Proposition \ref{prop:RMP}.

\begin{prop}
Let $W$ be a BMSC, $W^{\rm opt}$ be the perfect channel and $\alpha \in [0, 1]$. Denote by $C_{\alpha}$ the polar code of block length $N$ and rate $R$ designed for transmission over the BMSC $W_{\alpha} = \alpha W + (1-\alpha) W^{\rm opt}$. Then, when $\alpha \to 0$, $C_{\alpha}$ is an RM code.
\end{prop}

\begin{proof}
When transmission takes place over the BMSC $W_{\alpha}$, the Bhattacharyya parameter $Z_i(W_{\alpha})$ of the $i$-th synthetic channel $W_{\alpha, N}^{(i)}$ ($i \in \{0, \cdots, N-1\}$) has the form \eqref{eq:Bhatt}, where $\varepsilon$ is replaced by $Z_{\alpha} = \alpha Z$, $f_1(x) = x^2$, and $f_0(x)$ can be bounded as \cite{arikan:polar}
\begin{equation} \label{eq:boundf0}
x \le f_0(x) \le 2x-x^2.
\end{equation}
Suppose that $g_{j^*}$ is included in the generator matrix of the code, but not $g_{i^*}$, with $w_{\rm H}(g_{i^*}) > w_{\rm H}(g_{j^*})$. Then, using \eqref{eq:boundf0}, $Z_{i^*}$ can be upper bounded by a polynomial in $\alpha$ with minimum degree $w_{\rm H}(g_{i^*})$ and $Z_{j^*}$ can be lower bounded by a polynomial in $\alpha$ with minimum degree $w_{\rm H}(g_{j^*})$. Thus, for $\alpha$ small enough $Z_{i^*} < Z_{j^*}$ and we reach a contradiction.
\end{proof}

Remark that if $W=$ BEC$(\varepsilon)$, then $W_{\alpha}=$  BEC$(\alpha\varepsilon)$. In general, there might be more natural ways to obtain the family of codes ${\mathcal C}_{\rm inter}$, according to the particular choice of the channel $W$. Indeed, in Section \ref{subsec:BAWGNC} which deals with the case of the BAWGNC, the interpolating family is constructed in a different way.

Once obtained a family of codes of the form $C_{\alpha}$, where $C_{\alpha}\big |_{\alpha =1}$ is the polar code designed for transmission over the channel $W$ and $C_{\alpha}\big |_{\alpha =0}$ is an RM code, numerical simulations show that the error probability under MAP decoding is an increasing function of $\alpha$. On the other hand, under SC decoding, the optimal performance is still achieved using $C_{\alpha}\big |_{\alpha =1}$. If one considers low-complexity decoding algorithms which get close to the error probability under MAP decoding, the finite-length performance of polar codes is significantly improved by using the code $C_{\alpha}$ for a suitable choice of the parameter $\alpha$.

\subsection{Case Study: \texorpdfstring{$W =$ BAWGNC$(\sigma^2)$}{W = BAWGNC}} \label{subsec:BAWGNC}

Let $W$ be a binary additive white Gaussian noise channel with variance of the noise $\sigma^2$, in short $W=$ BAWGNC$(\sigma^2)$, and define $C_{\alpha}$ as the polar code designed for transmission over $W_{\alpha}=$ BAWGNC$(\alpha \sigma^2)$. As $\alpha \to 0$, $W_{\alpha}$ tends to the perfect channel $W^{\rm opt}$ and $C_{\alpha}$ becomes an RM code. In order to show the performance improvement guaranteed by the usage of codes in the interpolating family ${\mathcal C}_{\rm inter}$ defined as in \eqref{eq:Cinterp}, consider the SCL decoder. To be coherent with the simulation setup of \cite{vardy:listpolar}, the numerical simulations refer to codes of fixed block length $N= 2^{11}$ and rate $R=0.5$. The number of Monte Carlo trials is $M = 10^5$. The codes are optimized for an SNR = 2 dB, namely, $\sigma^2 = 0.6309$ (recall that SNR = $1/\sigma^2$). The results of Figure \ref{fig:BAWGNSCLdiffL} are qualitatively similar to those represented in Figure \ref{fig:SCLdiffL} for the BEC and testify the remarkable performance gain achievable by codes of the form $C_{\alpha}$ with respect to the original polar code $C_{\alpha}\big |_{\alpha = 1}$.

\begin{figure}[t]
\centering
\subfigure[$L=8$]
{\includegraphics[width=0.86\columnwidth]{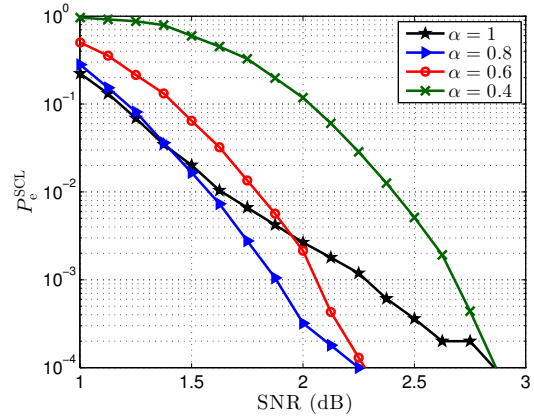}}\\
\subfigure[$L=32$]
{\includegraphics[width=0.86\columnwidth]{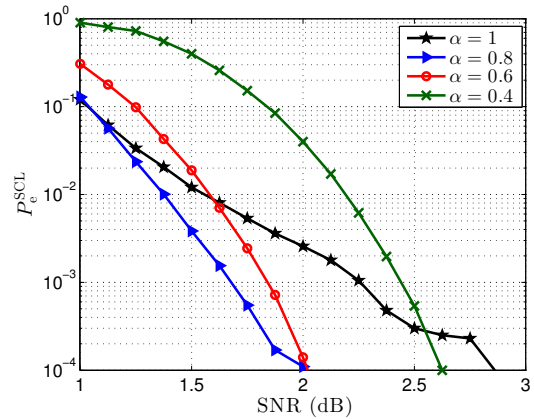}}
\caption{Error probability $P_{\rm e}^{\rm SCL}$ under SCL decoding for the transmission of $C_{\alpha}$ over the BAWGNC$(\sigma^2)$, where $\sigma^2 = 0.6309$, the SNR varies from $1$ to $3$ with a step of $0.125$ and $\alpha \in \{0.4, 0.6, 0.8, 1\}$. The block length is $N=2^{11}$ and the rate is $R=0.5$. For the target error probability $P_{\rm e} = 10^{-3}$ an improvement $\ge 0.5$ dB with respect to the original polar code $C_{\alpha}\big |_{\alpha=1}$ can be noticed using the code $C_{\alpha}\big |_{\alpha=0.8}$ when $L=32$.}
\label{fig:BAWGNSCLdiffL}
\end{figure}

\section{Concluding Remarks} \label{sec:concl}

As pointed out in \cite{vardy:listpolar}, the error probability of polar codes at practical block lengths can be reduced by acting both on the decoder and on the code itself. Unfortunately, an improvement only in the decoding algorithm does not seem to be enough to change the scaling exponent \cite{mondelli:scalingITW}. In this work we address the issue of boosting the finite-length performance of polar codes by modifying jointly the code and the SC decoding algorithm. In particular, we construct a family of codes ${\mathcal C}_{\rm inter} = \{C_{\alpha} : \alpha \in [0, 1]\}$ of fixed block length and rate which interpolates from the original polar code $C_{\alpha}\big |_{\alpha=1}$ to the RM code $C_{\alpha}\big |_{\alpha=0}$. Numerically, the error probability under MAP decoding decreases as $\alpha$ goes from 1 to 0. Since MAP decoding is not practical for transmission over general channels, we develop a trade-off between complexity and performance by considering low-complexity decoders (e.g., BP, SCL). As a result, we show the significant benefit coming from the adoption of codes in ${\mathcal C}_{\rm inter}$ via numerical simulations for the BEC and the BAWGNC. This improvement in the finite-length performance of polar codes can be substantial: we provide experimental evidence of the fact that the error probability under MAP decoding for the transmission over the BEC of $C_{\alpha}$ for $\alpha$ sufficiently small is very close to that of random codes, which achieve a better scaling exponent than polar codes. An interesting open question concerns the extension of the findings of this paper to non-binary channels by constructing polar codes with arbitrary input alphabet sizes \cite{STA09}.

\section*{Acknowledgement}

The authors would like to thank M. B. Parizi for providing the code which simulates the SCL decoder for transmission over the BAWGNC. M. Mondelli is supported by grant No. 200020\_146832/1 of the Swiss National Science Foundation.

\bibliographystyle{IEEEtran}

\bibliography{biblio}

\end{document}